\makeatletter
\let\@twosidetrue\@twosidefalse
\let\@mparswitchtrue\@mparswitchfalse
\makeatother%this is to shovel all comments to the right margin 
\documentclass{llncs}

\usepackage[utf8]{inputenc}
\usepackage[english]{babel}

\usepackage{booktabs} % For formal tables
\usepackage[ruled]{algorithm2e} % For algorithms
\usepackage{complexity}
\usepackage{graphicx}
\usepackage{geometry}
\usepackage{float}
\usepackage{xcolor,xspace,cite}
\usepackage{paralist}
\usepackage{refcount,amsfonts}
\usepackage{hyperref,lmodern}

\usepackage[textsize=tiny,textwidth=1.9cm,shadow]{todonotes}

\newtheorem{pr}{Problem}
\newcommand{\inp}{\textsf{Input: }} 

\newcommand{\ques}{\textsf{Question: }}

\usetikzlibrary{arrows,decorations.markings,patterns,calc, positioning, backgrounds, shapes}
\tikzstyle{vertex} = [circle, draw=black, scale=0.7]
\tikzstyle{edgelabel} = [circle, fill=white, scale=0.8, inner sep=1pt]

\tikzstyle{pointedline} = [line width=0.3mm,
decoration={markings,
	mark=at position 1 with {\arrow[line width=0.2mm, scale= 2]{>}}
},
postaction={decorate}
]
\tikzset{
  treenode/.style = {shape=rectangle, rounded corners,
                     draw, align=left,
                     top color=white, bottom color=blue!20},
  root/.style     = {treenode, bottom color=red!30},
  env/.style      = {treenode, font=\ttfamily\normalsize},
  dummy/.style    = {circle,draw}
}

\title{Pareto optimal coalitions of fixed size
\thanks{Supported by the Cooperation of Excellences Grant (KEP-6/2018), by the Ministry of Human Resources under its New National Excellence Programme (ÚNKP-18-4-BME-331), the Hungarian Academy of Sciences under its Momentum Programme (LP2016-3/2016), its J\'anos Bolyai Research Fellowship, and OTKA grant K128611. This work is connected to the scientific program of the ``Development of quality-oriented and harmonized R+D+I strategy and functional model at BME'' project, supported by the New Hungary Development Plan (Project ID:T\'{A}MOP-4.2.1/B-09/1/KMR-2010-0002).}} 
\author{\'{A}gnes Cseh\inst{1} \and Tam\'{a}s Fleiner\inst{1,2} \and Petra Harj\'{a}n\inst{2}}
\institute{
Hungarian Academy of Sciences, Centre for Economic and Regional Studies, Institute of Economics
\and
Department of Computer Science and Information Theory, Budapest University of Technology and Economics
}

\begin{document}
\maketitle

\begin{abstract}
We tackle the problem of partitioning players into groups of fixed size, such as allocating eligible students to shared dormitory rooms. Each student submits preferences over the other individual students. We study several settings, which differ in the size of the rooms to be filled, the orderedness or completeness of the preferences, and the way of calculating the value of a coalition---based on the best or worst roommate in the coalition. In all cases, we determine the complexity of deciding the existence, and then finding a Pareto optimal assignment, and the complexity of verifying Pareto optimality for a given assignment.
\end{abstract}

\section{Introduction}%\accom{Although the authors present their work as a roommate assignment setting, I feel it has a more definite hedonic game flavour, and the authors could consider rephrasing the paper in that way.} \accom{. Even the most central concepts are introduced informally only. This applies to, for instance, coalition formation games themselves, the concept of Pareto optimality, preferences based on best and worst players.}

The ubiquitous nature of coalition formation has been triggering researchers from various fields to study the behavior of individuals forming groups~\cite{DG80,Gam61,KR14,KC82,SK98}. A large portion of the game-theoretical studies focuses on the individuals having ordinal preferences over the possible outcomes~\cite{BKS01,BJ02,CH04a,CH04b,PE15}.

Our setting involves $n$ players who need to be partitioned into coalitions. For convenience, we talk about assigning each player to a room. We assume the rooms to have no specific feature besides their capacity. To ensure a feasible partition in the outcome, we assume that the total capacity of rooms adds up to $n$, and a feasible assignment fills each room to its capacity. Having once entered the scheme, players have no option to opt out. Each player submits her preference list over the other players. Four features of the problem define various settings, each of which is realistic and will be investigated by us.

\begin{enumerate}[(i)]
	\item \label{it:ci} Preference lists might be \emph{complete} or \emph{incomplete}, this latter meaning that players have the right to declare some of the other players unacceptable as a roommate. No player can be put in the same room with an unacceptable roommate. %A coalition is infeasible if a roommate declares at least one other roommate unacceptable (or, if the room is not filled up to its capacity, as we already mentioned).
	\item \label{it:sw} Players might have \emph{strictly} or \emph{weakly ordered} lists.
	\item \label{it:bw} Players might compare two coalitions based on their \emph{best} or \emph{worst} assigned roommate.
	\item \label{it:3i} The rooms might all accommodate 3 players, or they might have a predefined \emph{capacity} each, which we denote by $r_1, r_2, \dots, r_k$ for each of the $k$ rooms.
\end{enumerate}

The optimality principle we are focusing on is \emph{Pareto optimality}. A room assignment, or, a set of coalitions is Pareto optimal if there is no other assignment in which at least one player is better off, and no player is worse off than in the first assignment. The comparison here is defined based on point~(\ref{it:bw}) above. We shorten the term `Pareto optimal assignment' to \textsc{poa}. Our goal is to study all $2^4$ combinations of the above four features, and for each of them, determine the complexity of the following three problems:
\begin{enumerate}[(1)]
\item verifying whether a given feasible assignment is a \textsc{poa};
\item checking whether a \textsc{poa} exists;
\item finding a \textsc{poa}.
\end{enumerate}

\subsection{Related literature}

Pareto optimality in coalition formation has a rich literature. We start with an overview on coalition formation viewed as a hedonic game. Then we review the settings in which the coalition size matters.%\accom{The problem of coalition formation with fixed coalition sizes is natural enough, but the problem is not further motivated by the authors, neither conceptually nor technically. Also the concept of Pareto optimality is left without any comment as to its role and importance in the social sciences.}

\paragraph{Hedonic games.} Coalition formation under preferences can be seen as a hedonic game~\cite{ABS11,BKS01,BJ02}. In such a game, players have preferences over the possible coalitions they can be part of, and coalitions can be of any size---notice that these two basic features strikingly differ from our setting. Pareto optimal coalition formation as a hedonic game is extensively studied by Aziz et al.~\cite{ABH13}. They analyze two restricted variants of hedonic games that are closely related to our setting. Both of them operate under the assumption that the coalition size is arbitrary, but they both derive players' preferences on coalitions from a preference list on individual players. They show that if the preference lists are incomplete, and preferences are based on the best roommate, then verifying Pareto optimality for a given assignment is $\coNP$-complete, and computing a \textsc{poa} is $\NP$-hard. For complete lists and the same preferences, the grand coalition is a trivial optimal solution, since every player has their first-choice roommate in the sole room. Aziz et al.~\cite{ABH13} also show that if preferences are based on the worst roommate, then both verifying Pareto optimality for a given assignment, and computing a \textsc{poa} can be done in polynomial time.

\paragraph{2-person rooms.} Some of the literature concentrates on each coalition being of size 1 or~2. In this setting, a player can compare two coalitions simply based on the rank of the sole roommate (if any exists), so our point~(\ref{it:bw}) does not apply here. Using Morrill's algorithm~\cite{Mor10}, Aziz et al.~\cite{ABH13} show that even if preferences contain ties, both verifying Pareto optimality for a given matching, and calculating a \textsc{poa} are solvable in polynomial time. Their results are valid for complete and incomplete lists as well. Abraham and Manlove~\cite{AM04} consider Pareto optimal matchings as a means of coping with instances of the stable roommates problem with strict lists, which do not admit a stable matching. They show that while a maximum size \textsc{poa} is easy to find, finding a minimum size \textsc{poa} is $\NP$-hard.

\paragraph{3-person rooms.} For the setting in which a room can accommodate up to 3 players, two versions of the problem have been studied. A \emph{three-cyclic game} is a hedonic game in which the set of players is divided into men, women, and dogs and only kind of acceptable coalitions are man-woman-dog triplets~\cite{Knu76,NH91}. Furthermore, men only care about women, women only care about dogs and dogs only care about men. Computing a \textsc{poa} is known to be $\NP$-hard for these games, while the corresponding verification problem is $\coNP$-complete, even for strict preferences~\cite{ABH13}. In \emph{room-roommate games}, a set of players act as rooms, and these have no preferences whatsoever. The ordinary players, on the other hand, have a preference list over all possible roommate-room pairs they find acceptable. A triplet is feasible if exactly one player in it plays a room. If preferences are strict, then a \textsc{poa} can be computed in polynomial time, but the problem becomes $\NP$-hard as soon as ties are introduced, even if all lists are complete~\cite{ABH13}. Just as in the previous problem, the verification version is $\coNP$-complete, even for strict preferences~\cite{ABH13}.

\paragraph{Preferences depending on the room size.} Anonymous hedonic games~\cite{Bal04} are a subclass of hedonic games in which the players' preferences over coalitions only depend on coalition sizes. Both verification and finding a \textsc{poa} are hard in such games~\cite{ABH13}. Darmann~\cite{Dar18} studies a group activity selection model in which players have preferences not only on the activity, but also on the number of participants in their coalition. He provides an efficient algorithm to find a \textsc{poa}, if each player wants to share an activity with as many, or as few players as possible. 
\medskip

As we have seen, a number of papers investigate the complexity of finding a \textsc{poa} under various settings, such as limited coalition size or preferences over players rather than over coalitions. However, there is no work on the combination of these two. In our setting, homogeneous players \emph{rank each other}, and form coalitions of an arbitrary, but \emph{fixed} size. %Our paper makes an attempt to fill this gap in the literature.

\subsection{Our contribution}

We tackle the problems of verifying Pareto optimality, deciding the existence of \textsc{poa}, and finding a \textsc{poa} for a set of fixed coalition sizes. We distinguish $2^4$ cases, based on the completeness and the orderedness of preferences, the way of comparing two coalitions by a player, and the room sizes. Our findings are summarized below and in Figure~\ref{fi:results}. 
 
\begin{itemize}
	\item Verification is $\coNP$-complete in all cases. We show this in Section~\ref{se:verification} by two reductions from triangle cover problems.
    \item If lists are incomplete, then deciding whether a feasible assignment exists is already $\NP$-complete in all cases. On the other hand, if a feasible assignment does exist, then an optimal one exists as well. These are due to a simple $\NP$-completeness reduction and a monotonicity argument, which can be found in Section~\ref{se:existence}.
    \item For complete lists, a \textsc{poa} is bound to exist. In 3 out of the 16 cases, serial dictatorship delivers one. In all other cases, by computing \emph{any} \textsc{poa} in polynomial time, one could answer an $\NP$-complete decision problem in polynomial time. In Section~\ref{se:finding} we elaborate on these easy and hard cases. For the positive results, we interpret serial dictatorship in the current problem settings. Then, we utilize a tool developed by Aziz et al.~\cite{ABH13} in the hardness proofs.
\end{itemize}

%Section~\ref{se:verification} contains the proofs for the first point above, while Section~\ref{se:existence} is devoted to the second point. Finally, in Section~\ref{se:finding} we elaborate on the last point.

\begin{figure}
	\centering
\begin{tikzpicture}[grow=right,level distance=10em,edge from parent/.style = {draw, -latex}, every node/.style       = {font=\footnotesize}, sloped] 

  \tikzstyle{level 1}=[sibling distance=8em,level distance=15em] 
  \tikzstyle{level 2}=[sibling distance=8em,level distance=5em] 
  \tikzstyle{level 3}=[sibling distance=4em,level distance=7em] 
	\tikzstyle{level 4}=[sibling distance=2em,level distance=6em]

  \node [root] {Is there a \textsc{poa}?}
	child { node [treenode] {$\NP$-complete, Thm~\ref{th:3SI}}
	edge from parent node [below] {incomplete}}
	child { node [root] {Yes. How to find a \textsc{poa}?}
        child { node [dummy] {}
            child { node [treenode] {hard, Thm~\ref{th:3TCW}}
            edge from parent node [below, align=center]{worst}}
            child {  node [treenode] {hard, Thm~\ref{th:3TCB}}
            edge from parent node [above, align=center]{best}}
        edge from parent node [below] {ties} }
    child {node [dummy] {}
        child { node [treenode] {SD, Thm~\ref{th:3SCW}}
        edge from parent node [below, align=center]{worst}}
        child { node [dummy] {}
            child { node [treenode] {hard, Thm~\ref{th:iSCB}}
            edge from parent node [below, align=center]{$r_i$}}
            child { node [treenode] {SD, Thm~\ref{th:3SCB}}
            edge from parent node [above, align=center]{3}}
        edge from parent node [above, align=center]{best}}
        edge from parent node [above] {strict}}
    edge from parent node [above] {complete}};
\end{tikzpicture}
\caption{The complexity of finding a \textsc{poa}. The arrows explain the case distinctions according to points (\ref{it:ci})--(\ref{it:3i}) listed in the Introduction. SD stands for serial dictatorship.}
\label{fi:results}
\end{figure}
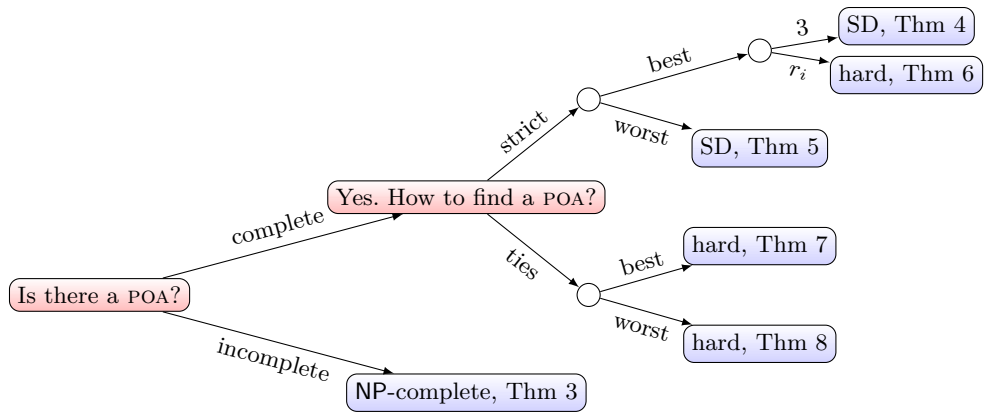

\section{Preliminaries}
\label{se:preliminaries}

In this section we set the solid mathematical basis for discussing our problems. Then, we introduce a selection of $\NP$-complete problems and prove that a specific variant  of them is hard as well. We do this to prepare the ground for our hardness proofs in Sections~\ref{se:verification}-\ref{se:finding}.

\subsection{Problem definition}

Our input consists of a set of players $P$, with $|P|=n$, a multiset $r_1, r_2, \ldots, r_k$ of room capacities with $r_1+r_2+\ldots +r_k = n$, and strictly ordered, but not necessary complete preferences for each player $i \in P$ on other players. An assignment $R$ is the partition of $P$ into sets $R_1, R_2, \ldots, R_k$. This means that players have no choice to opt out once they have entered the market. %Two assignments are equivalent if 

\begin{definition}
\label{def:assignment}
 An assignment $R=\left\{R_1, R_2, \ldots, R_k\right\}$ is called feasible if 
 \begin{enumerate}
     \item $|R_i| = r_i$ for each $1 \leq i \leq k$;
     \item each player in $R_i$ declares every other player in the same $R_i$ acceptable.
 \end{enumerate}
\end{definition}

Players evaluate their situation in an assignment solely based on the roommates they are grouped together with and do not care about the other rooms. For point~2 in Definition~\ref{def:assignment}, players only need to compare outcomes in which they find all other roommates acceptable. We define two comparison principles.

When the \emph{best} roommate counts, coalition $R_i$ is preferred to $R_i'$ by player $i$ if player $j \in R$ ranked highest by $i$ among all players in $R_i$ is preferred by $i$ to player $j' \in R$ ranked highest by $i$ among all players in $R_i'$. Analogously, when the \emph{worst} roommate counts, coalition $R_i$ is preferred to $R_i'$ by player $i$ if player $j \in R$ ranked lowest by $i$ among all players in $R_i$ is preferred by $i$ to player $j' \in R$ ranked lowest by $i$ among all players in $R_i'$. Our definitions are aligned with $\mathcal{W}$-preferences in~\cite{CH04b}, and with $B$-hedonic and $W$-hedonic games in~\cite{ABH13}. However, $\mathcal{B}$-preferences in~\cite{CH04b} are different, because there, the size of the coalition serves as a tiebreaker in case the best roommate is identical in the two coalitions. We consider two such coalitions equally good.

\subsection{Relevant hard problems}

We first introduce two variants of a hard graph cover problem that we will later use in our hardness reductions, and then shortly mention a bin packing problem as well.

\begin{pr}\textsc{ triangle cover} \\
	\inp A simple graph $G=(V,E)$.\\
	\ques Does there exist a partition of $V$ into sets of cardinality 3 so that for each set, the three vertices span a 3-cycle?
\end{pr}

\begin{pr}\textsc{directed triangle cover} \\
	\inp A simple directed graph $D=(V,A)$.\\
	\ques Does there exist a partition of $V$ into sets of cardinality 3 so that for each set, the three vertices span a directed 3-cycle?
\end{pr}

The \textsc{triangle cover} problem asks for a set of vertex-disjoint 3-cliques (triangles) in the input graph $G$, so that these cover all vertices of~$G$. This problem has been shown to be $\NP$-complete by Garey and Johnson~\cite{GJ79}. The directed version has not been proved to be $\NP$-complete, so we give a simple hardness proof below. %Our proof is valid even for simple graphs, meaning that two vertices are connected by at most one directed edge, and no vertex is connected to itself.
The hard problem we reduce to \textsc{directed triangle cover} is \textsc{3d hypergraph matching}~\cite{Kar72,PS82}.

\begin{pr}\textsc{3d hypergraph matching} \\
	\inp A hypergraph $H=(U\cup V\cup W,E)$, where each $e \in E$ is a triple $(u,v,w)$ so that $u \in U$, $v \in v$, and $w \in W$.\\
	\ques Does there exist a perfect matching in~$H?$
\end{pr}

\begin{claim}
      \textsc{directed triangle cover} is an $\NP$-complete problem.
\end{claim}

\begin{proof}
We start with the input graph of an arbitrary instance of \textsc{3d hypergraph matching} and transform it into an instance of \textsc{directed triangle cover}. We keep the set of vertices, and simply replace each hyperedge $(u,v,w)$ of $H$ by the gadget shown in Figure~\ref{fi:directed}. This gadget introduces 9 new vertices per hyperedge to the set of vertices, which we will call the \emph{gadget vertices}. 
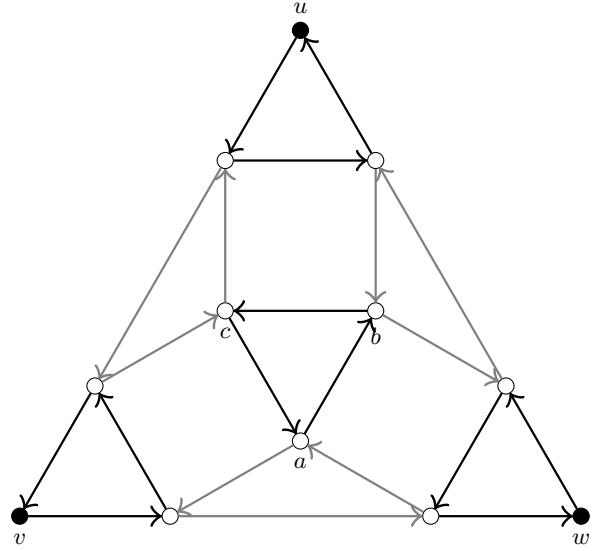
\begin{figure}
    \centering
\begin{tikzpicture}[scale=1, transform shape]
		%\tikzstyle{vertex} = [circle]
		\pgfmathsetmacro{\d}{2}
		\node[vertex, fill=black, label=below:$v$] (v) at (0, 0) {};
		\node[vertex, label=below:$$] (v1) at ($(v)+(\d,0)$) {};
		\node[vertex, label=below:$$] (v2) at ($(v1)+sqrt(3)*(\d,0)$) {};
		\node[vertex, fill=black, label=below:$w$] (w) at ($(v2)+(\d,0)$) {};
		
		\node[vertex, label=below:$$] (u2) at ($(v)+(0.5*\d, 0.866*\d)$) {};
		\node[vertex, label=below:$$] (u1) at ($2.732*(u2)$) {};
		\node[vertex, fill=black, label=above:$u$] (u) at ($3.732*(u2)$) {};
		
		\node[vertex, label=below:$$] (w2) at ($(u1)+(\d,0)$) {};
		\node[vertex, label=below:$$] (w1) at ($(u2)+(2.732*\d,0)$) {};
		
		\node[vertex, label=below:$c$] (c) at ($(u1)+(0,-\d)$) {};
		\node[vertex, label=below:$b$] (b) at ($(w2)+(0,-\d)$) {};
		\node[vertex, label=below:$a$] (a) at ($(v1)!0.5!(v2)+(0,0.5*\d)$) {};
		
		\draw [pointedline] (v) -- (v1);	
		\draw [pointedline] (v1) -- (u2);
		\draw [pointedline] (u2) -- (v);
		\draw [pointedline] (a) -- (b);
		\draw [pointedline] (b) -- (c);
		\draw [pointedline] (c) -- (a);
		\draw [pointedline] (u) -- (u1);
		\draw [pointedline] (u1) -- (w2);
		\draw [pointedline] (w2) -- (u);
		\draw [pointedline] (v2) -- (w);
		\draw [pointedline] (w) -- (w1);
		\draw [pointedline] (w1) -- (v2);
		
		\draw [pointedline,gray] (u1) -- (u2);
		\draw [pointedline,gray] (u2) -- (c);
		\draw [pointedline,gray] (c) -- (u1);
		\draw [pointedline,gray] (w1) -- (w2);
		\draw [pointedline,gray] (w2) -- (b);
		\draw [pointedline,gray] (b) -- (w1);
		\draw [pointedline,gray] (v1) -- (v2);
		\draw [pointedline,gray] (v2) -- (a);
		\draw [pointedline,gray] (a) -- (v1);
\end{tikzpicture}
    \label{fi:directed}
    \caption{The gadget replacing each hyperedge $(u,v,w) \in E(H)$. The unfilled vertices are gadget vertices, and these are not connected to any vertex outside of this gadget. Triangles in this gadget are either black or gray. The gray triangles cover all of these, and leave the remaining three vertices uncovered, while the black triangles cover all vertices in the gadget.}
\end{figure}

Assume first that a directed triangle cover exists in the resulting directed graph~$D$. For each gadget, two scenarios can occur. Either vertices $a,b,c$ form a triangle and then all of $u, v$, and $w$ must be covered by triangles inside the same gadget, or triangle $(a,b,c)$ is not in the cover, but then all of $u, v$, and $w$ are covered by triangles outside the gadget. The first case corresponds to the hyperedge $(u,v,w)$ being in our perfect matching $M$, while the second case tells otherwise. We still need to show that $M$ is indeed a perfect matching. Firstly, no vertex can appear in two hyperedges in $M$, because it would cover the same vertex twice by directed triangles. Secondly, if a vertex is left unmatched, then it is also left uncovered by the triangle cover, which is a contradiction.

To show the opposite direction, we first assume that a perfect matching $M$ exists in the original instance. We transform it into directed triangles in the created directed graph $D$ as above: the edges of $M$ will be the gadget with $(a,b,c)$ (the black triangles), the rest of the edges will be the gadgets without $(a,b,c)$ (the gray triangles). We need to show that these triangles form a cover. The gadget vertices are trivially covered exactly once. The rest of the vertices are matched in $M$ along some hyperedge, and the gadget belonging to this hyperedge covers them with a triangle inside the gadget.
%a pair of parallel edges $\left\{(u,v),(v,u)\right\}$ pointing to opposite directions to derive the directed graph~$D$, which will be the input of our \textsc{directed triangle cover} problem. If a triangle cover exists in $G$, then by taking an arbitrary orientation of the corresponding triangles we can construct a directed triangle cover in~$D$. If a directed triangle cover exists in $D$, then it must originate from a triangle cover in~$G$.%\accom{jelölés?}
\qed \end{proof}

\begin{remark}
\label{re:triangle}
    It is an obvious observation that both \textsc{triangle cover} and \textsc{directed triangle cover} are hard only in graphs for which $|V|$ is a multiple of~3. We will thus assume that $G$ in the input of these hard problems has this property. Besides this, we can also assume that $G$ is simple, and that each vertex is of degree at least~2, because graphs with an isolated or a degree~1 vertex are trivial no-instances. For the directed version, one can even assume that each vertex has at least one outgoing and at least one incoming edge.
\end{remark}

\medskip

We close this section with the final element in our toolbox of hard problems. The input of \textsc{unary bin packing} is a set of item sizes, and a bin size, all encoded in unary. The goal is to group all items into bins so that the total item size in each bin is exactly the bin size. This problem has been shown to be $\NP$-complete by Garey and Johnson~\cite{GJ79}.

\begin{pr}\textsc{unary bin packing} \\
	\inp A set of item sizes $i_1, i_2, \ldots, i_n$, and a bin size $b$, all encoded in unary.\\
	\ques Does there exist a partitioning of $i_1, i_2, \ldots, i_n$ so that the sum of item sizes in each set of the partition adds up to~$b$?
\end{pr}

\begin{remark}
For our proofs, we will assume that the smallest item size is at least~2. The hardness of this variant is easy to see. If we take an input of \textsc{unary bin packing} and multiply all item and bin sizes by 2, then we get an equally hard problem that can be encoded in twice as many bits as the original one.
\label{re:bin}
\end{remark}

\section{Verification}%\accom{I feel that the Section 3 on deciding the existence of Pareto optimal and individual rational outcomes should come first, that is, before the sections on verification and finding Pareto optimal assignments.}
\label{se:verification}

In this section, we show the hardness of verification for all cases. We present two proofs: in Theorem~\ref{th:verW}, the worst roommate defines the base of comparison for two coalitions, while it is the best roommate who counts in Theorem~\ref{th:verB}. Other than this, we restrict our reduction to the least general case of the problem, having strict and complete lists, and 3-person rooms.

\begin{theorem}[Verification, strict and complete lists, worst roommate counts, 3-person rooms]
\label{th:verW}
Supposing that the preferences of a player depend on the worst roommate, verifying whether a given assignment is Pareto optimal is a $\coNP$-complete task even if all preferences are strict and complete, and every room is of size~3.
\end{theorem}

\begin{proof}
To each instance $G$ of \textsc{triangle cover} we will construct an instance of our verification problem, consisting of players with strict and complete preferences, and an assignment on which Pareto optimality is to be verified. We will show that a triangle cover exists in the first instance if and only if the assignment has a Pareto improvement.

\begin{figure}[ht]
\centering
\begin{tikzpicture}[scale=1,yscale=-1,rotate=90]
\pgfmathsetmacro{\b}{240}
\pgfmathsetmacro{\d}{120}
\pgfmathsetmacro{\a}{20}
\pgfmathsetmacro{\c}{160}

\draw [draw=black] (0.8,10) rectangle (0,0);
\draw [draw=black] (-3,10) rectangle (-3.8,0);
\draw [draw=black] (4.6,10) rectangle (3.8,0);

\draw[color=gray,dashed, ultra thick] (0.4,9.5) --  (4.2,9.5);
\draw[color=gray,dashed, ultra thick] (0.4,9.5) -- node[edgelabel,  near end]   {$3$}  (-3.4,9.5);
\draw [color=gray, dashed, ultra thick] (4.2,9.5) to[out=\c,in=\a] node[edgelabel,very  near end]   {$4$} (-3.4,9.5);
\draw[color=gray,dashed, ultra thick] (0.4,8) -- node[edgelabel,  near start]   {$3$} node[edgelabel, near end]   {$4$} (4.2,8);
\draw[color=gray,dashed, ultra thick] (0.4,8) -- node[edgelabel,  near start]   {$4$} node[edgelabel,very  near end]   {$4$}(-3.4,8);
\draw [color=gray, dashed, ultra thick] (4.2,8) to[out=\c,in=\a] node[edgelabel,very  near end] {$5$} node[edgelabel,very  near start]   {$3$}  (-3.4,8);
\draw[color=gray,dashed, ultra thick] (0.4,6.5) -- (4.2,6.5);
\draw[color=gray,dashed, ultra thick] (0.4,6.5) -- (-3.4,6.5);
\draw [color=gray, dashed, ultra thick] (4.2,6.5) to[out=\c,in=\a] (-3.4,6.5);
\draw[color=gray,dashed, ultra thick] (0.4,5) -- (4.2,5);
\draw[color=gray,dashed, ultra thick] (0.4,5) -- (-3.4,5);
\draw [color=gray, dashed, ultra thick] (4.2,5) to[out=\c,in=\a] (-3.4,5);

\draw[color=blue,dotted, ultra thick] (0.4,9.5) --node[edgelabel,  near end]   {1} (0.4,8);
\draw[color=blue,dotted, ultra thick] (4.2,9.5) -- node[edgelabel,  near end]   {1} node[edgelabel,near start]   {$2$}(4.2,8);
\draw[color=blue,dotted, ultra thick] (0.4,9.5) -- (4.2,6.5);
\draw[color=blue,dotted, ultra thick] (0.4,8) -- node[edgelabel,  near start]   {2}(4.2,6.5);
\draw[color=blue,dotted, ultra thick] (0.4,6.5) -- (4.2,5);
\draw[color=blue,dotted, ultra thick] (0.4,6.5) -- (4.2,3.5);
%\draw[color=blue,dotted, ultra thick] (0.4,9.5) -- (0.4,8);
\draw[color=blue,dotted, ultra thick] (4.2,3.5) -- (4.2,5);
\draw[color=blue,dotted, ultra thick] (0.4,5) -- (0.4,3.5);
\draw[color=blue,dotted, ultra thick] (0.4,5) -- (2,3.4);
\draw[color=blue,dotted, ultra thick] (0.4,3.5) -- (1.8,2.8);
\draw[color=blue,dotted, ultra thick] (4.2,9.5) --  (2,10.5);
\draw[color=blue,dotted, ultra thick] (4.2,8) --  node[edgelabel,  near start]   {2}  (1.8,10);

\draw [color=black,thick] (-3.4,9.5)   to[out=\b,in=\d]  node[edgelabel, very near end]   {1} (-3.4,6.5);
\draw [color=black,thick] (-3.4,8) to[out=\b+10,in=\d]   (-3.4,5);
\draw [color=black,thick] (-3.4,8) to[out=\b,in=\d]   (-3.4,3.5);
\draw [color=black,thick] (-3.4,6.5) to[out=\b,in=\d]  node[edgelabel,  near start]   {2}   (-3.4,5);
\draw [color=black,thick] (-3.4,3.5) to[out=\b,in=\d]   (-3.5,2.7);
\draw [color=black,thick] (-3.4,0.5) to[out=\d,in=\b]   (-3.4,2.5);
\draw [color=black,thick] (-3.4,0.5) to[out=\d,in=\b]   (-3.4,1);
\draw [color=black,thick] (-3.4,9.5) to[out=\b,in=\d] node[edgelabel,  near start]   {12}  (-3.4,8);

\draw[fill=black] (0.4,9.5) circle (2pt);
\draw[fill=black] (0.4,8) circle (2pt);
\draw[fill=black] (0.4,6.5) circle (2pt);
\draw[fill=black] (0.4,5) circle (2pt);
\draw[fill=black] (0.4,3.5) circle (2pt);
\draw[fill=gray] (0.4,1.9) circle (1pt);
\draw[fill=gray] (0.4,1.6) circle (1pt);
\draw[fill=gray] (0.4,1.3) circle (1pt);
\draw[fill=black] (0.4,0.5) circle (2pt);

\draw[fill=black] (4.2,9.5) circle (2pt);
\draw[fill=black] (4.2,8) circle (2pt);
\draw[fill=black] (4.2,6.5) circle (2pt);
\draw[fill=black] (4.2,5) circle (2pt);
\draw[fill=black] (4.2,3.5) circle (2pt);
\draw[fill=gray] (4.2,1.9) circle (1pt);
\draw[fill=gray] (4.2,1.6) circle (1pt);
\draw[fill=gray] (4.2,1.3) circle (1pt);
\draw[fill=black] (4.2,0.5) circle (2pt);

\draw[fill=black] (-3.4,9.5) circle (2pt);
\draw[fill=black] (-3.4,8) circle (2pt);
\draw[fill=black] (-3.4,6.5) circle (2pt);
\draw[fill=black] (-3.4,5) circle (2pt);
\draw[fill=black] (-3.4,3.5) circle (2pt);
\draw[fill=gray] (-3.4,1.9) circle (1pt);
\draw[fill=gray] (-3.4,1.6) circle (1pt);
\draw[fill=gray] (-3.4,1.3) circle (1pt);
\draw[fill=black] (-3.4,0.5) circle (2pt);

\node at (0.4,-0.5) {$V_2$};
\node at (-3.4,-0.5) {$V_1$};
\node at (4.2,-0.5) {$V_3$};
\node at (-3.4,10.5) {$G$};
\end{tikzpicture}
\caption{The assignment instance constructed to~$G$ in the proof of Theorem~\ref{th:verW}. The numbers on the edges mark the preferences of the players. Players in $V_1$ rank other players in $V_1$ higher than players in $V_2$ and~$V_3$, as the numbers on their solid black and dashed gray edges indicate. Players in $V_2$ and $V_3$ prefer their dotted blue edges to their dashed gray edges. The dotted blue edges form a triangle cover of $V_2 \cup V_3$---they connect the leftmost two vertices in $V_3$ with the rightmost vertex in~$V_2$.}
\label{fi:verW}
\end{figure}
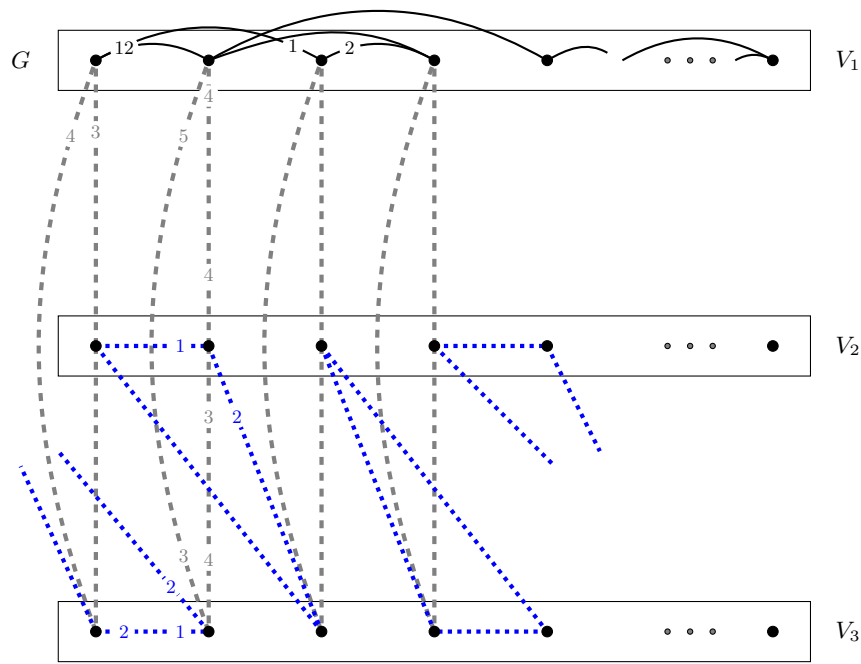

First we draw graph $G$ and also make two further copies of its vertex set~$V_1$. We denote these copies by $V_2$ and~$V_3$. Each vertex in $V_1 \cup V_2 \cup V_3$ represents a player in our \textsc{poa} instance. %Remember our remark from Section~\ref{se:preliminaries} implying that $|V_1| = |V_3| = |V_3|$ must be divisible by 3 for \textsc{triangle cover} to be hard in~$G$.
To show hardness in the most general case, we construct our \textsc{poa} instance with complete preference lists, which translates into a complete graph on the vertex set $V_1 \cup V_2 \cup V_3$. The edges of this graph can be partitioned into four classes.
\begin{itemize}
	\item The original edges of $G$ are solid and black in Figure~\ref{fi:verW}. These edges are the best choices of both of their end vertices, and the order among them can be chosen arbitrarily. Recall that each vertex in $V_1$ has at least two of these black edges, as argued in Remark~\ref{re:triangle} in Section~\ref{se:preliminaries}.
	\item The copied vertices $V_2$ and $V_3$ are connected by dotted blue triangles, as shown in Figure~\ref{fi:verW}. Each such triangle connects three vertices that originate from three different vertices of~$G$. Notice that these triangles can cover the entire set $V_2 \cup V_3$, because both $|V_2|$ and $|V_3|$ are multiples of~3, as noted in Remark~\ref{re:triangle}. The rank of these edges is the highest possible, and their order among themselves does not matter. 
	\item The three copies of the same vertex in $G$ are connected by a dashed gray triangle in Figure~\ref{fi:verW}. These edges are ranked lower than the solid black and the dotted blue edges. The order among them does not matter.
	\item All edges not visible in Figure~\ref{fi:verW} are ranked lower than the listed edges, in an arbitrary order.
\end{itemize}

The verification will happen with respect to the assignment $R$ built by all dashed gray edges.

If $G$ has a triangle cover, then these solid black triangles cover the entire set~$V_1$. In our assignment problem, these triangles translate into coalitions of size~3. The dotted blue triangles on $V_2$ and $V_3$ complete the alternative assignment. It is easy to see that every agent is better off by switching to the solid black and the dotted blue edges, since they are always ranked higher than the dashed gray edges. Thus, $R$ is not a \textsc{poa}.

Now we show the opposite direction. If there is a Pareto improvement to the assignment marked by the dashed gray edges, then it may use none of the invisible edges, since they are all worse than the dashed gray ones, and no player is allowed to receive a worse partner in the improved assignment. So we need to find an alternative assignment using only the first 3 types of edges. It is clear that by breaking any dashed gray coalition, at least one agent in each vertex group should be reassigned to a different coalition. Now, players in $V_2 \cup V_3$ can only choose a dotted blue edge instead of a dashed gray one. The sparsity of edges gives us that if there exists one dotted blue edge in a triangle, then not only the whole triangle is dotted blue, but we also must have all of the dotted blue triangles in the new assignment---otherwise players remain unassigned. Thus, in the new assignment, we must have all of the dotted blue triangles, and $V_1$'s vertices are free to be grouped up among themselves. This they can only do using the solid black edges, which correspond to the original edges in~$G$. Thus if we have a Pareto improvement of the dashed gray assignment, then we can cover $V_1$ with disjoint, solid black triangles. It means that $G$ has a triangle cover.
\qed \end{proof}

\begin{theorem}[Verification, strict and complete lists, best roommate counts, 3-person rooms]
\label{th:verB}
Supposing that the preferences of a player depend on the best roommate, verifying whether a given assignment is Pareto optimal is a $\coNP$-complete task even if all preferences are strict and complete, and every room is of size~3.
\end{theorem}% \accom{The argument of NP-completeness of DIRECTED TRIANGLE COVER depends on the possibility of such bidirected edges. The proof of Theorem 2, however, seems to assume that no such double edges are allowed. For suppose there are directed edges (v,w) and (w,v), how should v rank w, first or fifth? If the former, consider the graph (V,E) with V={a,b,c} and E={! (a,b),(b,a),(c,a)}. Then the edge (a,b) is ranked first by a, the edge (b,a) first by b, and the edge (c,a) first by c. But then, {{a,b,c},{a',b',c''},{a'',b'',c'}} seems to be a Pareto improvement over the coalition structure induced by the "red triangles" even though there is no triangle cover in the original graph. }

\begin{proof}
This proof follows the lines of the previous one, but it reduces our problem to \textsc{directed triangle cover}. To each instance $D$ of \textsc{directed triangle cover} we will construct an instance of our verification problem, consisting of players with preferences, and an assignment on which Pareto optimality is to be verified. We will show that a directed triangle cover exists in the first instance if and only if the assignment has a Pareto improvement.
\begin{figure}[ht]
\centering
\begin{tikzpicture}[scale=1,yscale=-1,rotate=90]
\pgfmathsetmacro{\b}{240}
\pgfmathsetmacro{\d}{120}
\pgfmathsetmacro{\a}{20}
\pgfmathsetmacro{\c}{160}

\draw [draw=black] (0.8,10) rectangle (0,0);
\draw [draw=black] (-3,10) rectangle (-3.8,0);
\draw [draw=black] (4.6,10) rectangle (3.8,0);

\draw[color=gray,dashed, ultra thick] (0.4,9.5) -- (4.2,9.5);
\draw[color=gray,dashed, ultra thick] (0.4,9.5) -- (-3.4,9.5);
\draw [color=gray, dashed, ultra thick] (4.2,9.5) to[out=\c,in=\a] (-3.4,9.5);
\draw[color=gray,dashed, ultra thick] (0.4,8) -- node[edgelabel,  near start]   {$2$} node[edgelabel, near end]   {$3$} (4.2,8);
\draw[color=gray,dashed, ultra thick] (0.4,8) -- node[edgelabel,  near start]   {$3$} node[edgelabel,very  near end]   {$2$}(-3.4,8);
\draw [color=gray, dashed, ultra thick] (4.2,8) to[out=\c,in=\a] node[edgelabel,very  near start]   {$2$}  node[edgelabel,very  near end]   {$3$}  (-3.4,8);
\draw[color=gray,dashed, ultra thick] (0.4,6.5) -- (4.2,6.5);
\draw[color=gray,dashed, ultra thick] (0.4,6.5) -- (-3.4,6.5);
\draw [color=gray, dashed, ultra thick] (4.2,6.5) to[out=\c,in=\a] (-3.4,6.5);
\draw[color=gray,dashed, ultra thick] (0.4,5) -- (4.2,5);
\draw[color=gray,dashed, ultra thick] (0.4,5) -- (-3.4,5);
\draw [color=gray, dashed, ultra thick] (4.2,5) to[out=\c,in=\a] (-3.4,5);

\draw[color=blue,dotted, ultra thick] (0.4,9.5) -- node[edgelabel,  near end]   {4} node[edgelabel,near start]{1} (0.4,8);
\draw[color=blue,dotted, ultra thick] (4.2,9.5) -- node[edgelabel,  near end]   {1} node[edgelabel,near start]   {4}(4.2,8);
\draw[color=blue,dotted, ultra thick] (0.4,9.5) -- node[edgelabel,  near end]   {1} node[edgelabel,very near start]{4}  (4.2,6.5);
\draw[color=blue,dotted, ultra thick] (0.4,8) --   node[edgelabel,  near end]   {4} node[edgelabel,near start]{1}(4.2,6.5);
\draw[color=blue,dotted, ultra thick] (0.4,6.5) -- (4.2,5);
\draw[color=blue,dotted, ultra thick] (0.4,6.5) -- (4.2,3.5);
\draw[color=blue,dotted, ultra thick] (4.2,3.5) -- (4.2,5);
\draw[color=blue,dotted, ultra thick] (0.4,5) -- (0.4,3.5);
\draw[color=blue,dotted, ultra thick] (0.4,5) -- (2,3.4);
\draw[color=blue,dotted, ultra thick] (0.4,3.5) -- (1.8,2.8);
\draw[color=blue,dotted, ultra thick] (4.2,9.5) -- (2,10.5);
\draw[color=blue,dotted, ultra thick] (4.2,8) -- node[edgelabel,  near start]   {4}  (1.8,10);

\begin{scope}[very thick,decoration={
    markings,
    mark=at position 0.5 with {\arrow{>}}}
    ] 
\draw [postaction={decorate},color=black,thick] (-3.4,9.5) to[out=\b,in=\d]  node[edgelabel, near end]   {$5$}   (-3.4,8);
\draw [postaction={decorate},color=black,thick] (-3.4,8) to[out=\b+10,in=\d]  node[edgelabel,very  near start]   {$1$}   (-3.4,5);
\draw [postaction={decorate},color=black,thick] (-3.4,3.5) to[out=\b,in=\d]  (-3.5,2.7);
\draw [postaction={decorate},color=black,thick] (-3.4,0.5) to[out=\d,in=\b]   (-3.4,2.5);
\draw [postaction={decorate},color=black,thick] (-3.4,0.5) to[out=\d,in=\b]   (-3.4,1);
\end{scope}
\begin{scope}[very thick,decoration={
    markings,
    mark=at position 0.5 with {\arrow{<}}}
    ]
\draw [postaction={decorate},color=black,thick] (-3.4,9.5) to[out=\b,in=\d]   (-3.4,6.5);
\draw [postaction={decorate},color=black,thick] (-3.4,8) to[out=\b-10,in=\d+10]  node[edgelabel, near start]   {$4$} node[edgelabel, near end]   {1}  (-3.4,3.5);
\draw [postaction={decorate},color=black,thick] (-3.4,6.5) to[out=\b,in=\d]   (-3.4,5);
\end{scope}

\draw[fill=black] (0.4,9.5) circle (2pt);
\draw[fill=black] (0.4,8) circle (2pt);
\draw[fill=black] (0.4,6.5) circle (2pt);
\draw[fill=black] (0.4,5) circle (2pt);
\draw[fill=black] (0.4,3.5) circle (2pt);
\draw[fill=gray] (0.4,1.9) circle (1pt);
\draw[fill=gray] (0.4,1.6) circle (1pt);
\draw[fill=gray] (0.4,1.3) circle (1pt);
\draw[fill=black] (0.4,0.5) circle (2pt);

\draw[fill=black] (4.2,9.5) circle (2pt);
\draw[fill=black] (4.2,8) circle (2pt);
\draw[fill=black] (4.2,6.5) circle (2pt);
\draw[fill=black] (4.2,5) circle (2pt);
\draw[fill=black] (4.2,3.5) circle (2pt);
\draw[fill=gray] (4.2,1.9) circle (1pt);
\draw[fill=gray] (4.2,1.6) circle (1pt);
\draw[fill=gray] (4.2,1.3) circle (1pt);
\draw[fill=black] (4.2,0.5) circle (2pt);

\draw[fill=black] (-3.4,9.5) circle (2pt);
\draw[fill=black] (-3.4,8) circle (2pt);
\draw[fill=black] (-3.4,6.5) circle (2pt);
\draw[fill=black] (-3.4,5) circle (2pt);
\draw[fill=black] (-3.4,3.5) circle (2pt);
\draw[fill=gray] (-3.4,1.9) circle (1pt);
\draw[fill=gray] (-3.4,1.6) circle (1pt);
\draw[fill=gray] (-3.4,1.3) circle (1pt);
\draw[fill=black] (-3.4,0.5) circle (2pt);

\node at (0.4,-0.5) {$V_2$};
\node at (-3.4,-0.5) {$V_1$};
\node at (4.2,-0.5) {$V_3$};
\node at (-3.4,10.5) {$D$};
\end{tikzpicture}
\caption{The assignment instance constructed to the directed graph~$D$ in the proof of Theorem~\ref{th:verB}. The numbers on the edges mark the preferences of the players. Players in $V_1$ rank their outgoing edges first, then their edges to players in $V_2$ and~$V_3$, and then their incoming edges, as the numbers on their solid black and dashed gray edges indicate. Preferences in the dashed gray and dotted blue triangles are cyclic. Players in $V_2$ and $V_3$ rank their two dashed gray edges sandwiched between their two dotted blue edges.}
\label{fi:verB}
\end{figure}
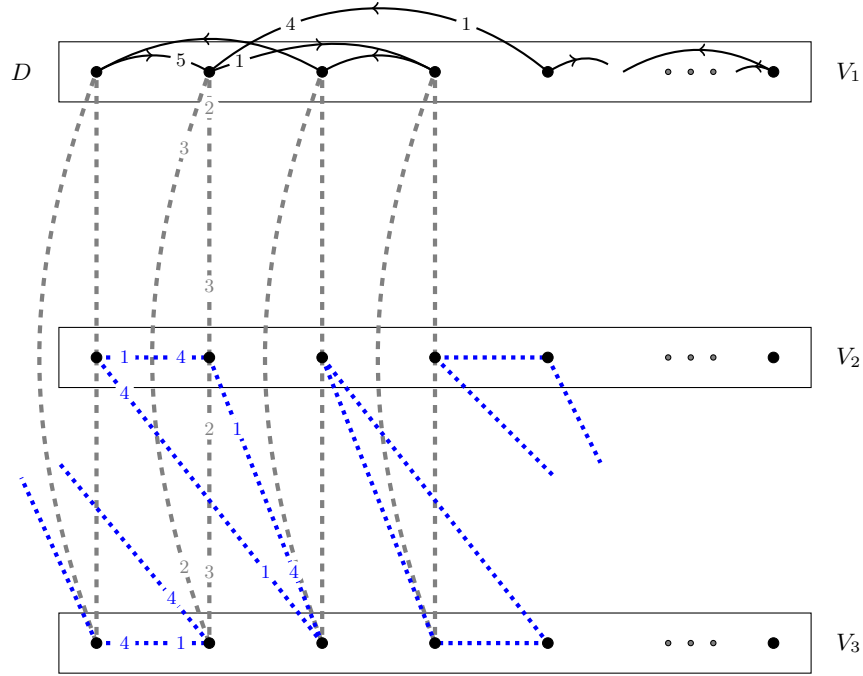

First we draw the directed graph $D$ and make two further copies of its vertex set~$V_1$. We denote these copies by $V_2$ and~$V_3$. Just as in our previous proof, the edges of this complete graph can be partitioned into four classes. Notice that the preferences differ from the ones in our previous proof.
\begin{itemize}
	\item The original directed edges of $D$ are solid and black in Figure~\ref{fi:verB}. These edges are the best choices of their starting vertex and they are ranked below the dashed gray edges at their end vertex. The order among all outgoing and among all incoming directed edges of the same vertex can be chosen arbitrarily. Recall that each vertex in $V_1$ has at least one outgoing and at least one incoming edge, as stated in Remark~\ref{re:triangle}.
	\item The copied vertices $V_2$ and $V_3$ are connected by dotted blue triangles, as shown in Figure~\ref{fi:verB}. Each such triangle connects three vertices that originate from three different vertices of~$D$. The rank of these edges is either 1 or 4, so that the triangle forms a preference cycle, i.e.~each edge is ranked first by one end vertex and fourth by the other one. Edges in the middle, ranked second and third, will be the dashed gray edges.
	\item The three copies of the same vertex in $G$ are connected by a dashed gray triangle in Figure~\ref{fi:verB}. These edges are ranked between outgoing and incoming edges in $D$ at vertices in $V_1$, and they are ranked second and third by vertices in $V_2 \cup V_3$. The order among them matters: the triangles themselves must form a preference cycle.
	\item All edges not visible in Figure~\ref{fi:verB} are ranked lower than the listed edges, in an arbitrary order.
\end{itemize}
Again, the verification will happen with respect to the assignment $R$ built by all dashed gray edges.

Suppose that there exists a \textsc{directed triangle cover} in~$D$. This means that each vertex in $V_1$ has an outgoing edge in the triangle cover, which is a first-choice roommate in the assignment problem. So switching to the coalitions marked by the triangle partition would be a Pareto improvement for the vertices in $V_1$. We need to take care of the vertices in $V_2 \cup V_3$ too. The dotted blue triangles complete the alternative assignment, and due to the cyclic nature of the preferences on them, they too assign each player a first choice roommate, which was not present in the original assignment.

Now suppose that there exists a Pareto improvement to the assignment $R$ marked by the dashed gray edges. In $R$, the best roommate of each player is ranked right below its outgoing black edges for vertices in~$V_1$, and second for vertices in $V_2$ and~$V_3$. In order to make at least one player, say, player $i$ better off, the alternative assignment $R'$ must allocate $i$ to a more preferred neighbor~$j$. This $j$ must be the end vertex of the directed solid black edge $(i,j)$, if $i \in V_1$, and $i$'s first choice roommate (dotted blue edge) if $i \in V_2 \cup V_3$. This is only possible if we find a (non-directed) triangle partition in the constructed graph so that it contains at least one solid black or dotted blue edge, which is~$(i,j)$. We now search for a third player to complete the coalition in~$R'$. Since $j$ just received a bad roommate in the person of $i$, she must have her first or second choice in the room as well, in order to keep her satisfied. Her second choice is one of her copies $j'$, but $j'$ only ranks $j$ third, and $i$ beyond all listed players, so her situation would worsen if we put her up in a room with $i$ and~$j$. On the other hand, $j$'s first choice edge $(j,\ell)$ is of the same color as $(i,j)$. The position of $\ell$ is not worse than in the original assignment if and only if $(i,j,\ell)$ forms a blue or a directed black triangle. This shows that only red, blue, or directed black triangles can appear in a Pareto improved assignment. The existence of such a monochromatic triangle partition implies that there is a directed triangle cover in~$D$. %TThe only way to assign every player to a room is to include all doted blue triangles, and find a directed triangle cover to assign all players in $V_1$ as well. 
\qed \end{proof}

\section{Existence}
\label{se:existence}

In all investigated cases, if the instance admits a feasible assignment, then it also admits a Pareto optimal one. Due to monotonicity in the rank of the worst or best roommate, a chain of Pareto improvements starting at any feasible assignment must end in a \textsc{poa}, which is thus guaranteed to exist if a feasible assignment exists. This is the case if lists are complete, since all assignments filling up all rooms are feasible then. Not so for incomplete lists---by declaring some other players unacceptable, players can easily reach a situation where not even feasible assignments exist. We now show that in all cases with incomplete lists, deciding whether a feasible assignment, and thus, a \textsc{poa} exists, is $\NP$-complete. Notice that since feasibility is already hard, it bears no importance whether a player judges the coalition based on the best or the worst roommate.

\begin{theorem}[Existence, strict and incomplete lists, best or worst roommate counts, 3-person rooms]
\label{th:3SI}
If lists are incomplete, then deciding whether a feasible assignment exists is $\NP$-complete even if all preferences are strict, and every room is of size~3.
\end{theorem}
\begin{proof}
We show hardness via a reduction from \textsc{triangle cover}. Given a graph $G$ as the input of this problem, we associate players with vertices and acceptable roommate pairs with edges. The preferences within the set of acceptable agents can be chosen arbitrarily, because they play no role in feasibility anyway. Since feasible assignments must form coalitions of size exactly~3, each assignment in our problem corresponds to a triangle cover and vice versa.
\qed \end{proof} 

This theorem shows that it is computationally infeasible to find a \textsc{poa} if lists are incomplete, because even deciding whether any exists is $\NP$-complete. Moreover, even if we are given a feasible assignment---which guarantees the existence of a \textsc{poa}---it is also computationally infeasible to find a Pareto improvement, since deciding whether the given assignment itself is a \textsc{poa} is $\coNP$-complete, as shown in Theorems~\ref{th:verW} and~\ref{th:verB}. From this point on, we can thus restrict our attention to instances with complete lists. %\accom{This is the conclusion of the the section and seems to contain the most important formal statement. Wouldn't this result warrant to be presented in a theorem-like environment? }

\section{Finding a Pareto optimal assignment}
\label{se:finding}

As already mentioned in the previous section, a \textsc{poa} is guaranteed to exist if lists are complete. Here we distinguish all $2^3$ cases based on three basic features of the problem, listed as points~(\ref{it:sw})--(\ref{it:3i}) in the Introduction. 

\subsection{Easy cases}%\accom{ One could wonder though, whether the serial dictatorship algorithms proposed here coincide with the Preference Refinement Algorithm of Aziz et al. [2] in the serial dictatorship setting. Some comment seems to be in order.}

We start with describing and analyzing the tailored variants of serial dictatorship in the cases where it delivers a \textsc{poa}. As usual, the core of serial dictatorship is that each dictator specifies a set of solutions that guarantee her one of her most desirable outcomes, and all later dictators must choose similarly, but within the already specified set. We show here algorithms that implement this principle under the problem settings we study.

Our algorithms are simple, and suited for the very specific problems in question, whereas the Preference Refinement Algorithm (\textsc{PRA}) of Aziz et al.~\cite{ABH13} for computing a Pareto optimal and individually rational assignment in hedonic games is generic, but also more complicated. Besides this, it invokes an oracle solving a
problem they call 'Perfect Partition'. Perfect Partition asks for an assignment that gives each player one of her best outcomes. (Notice that in this paper, we use the term perfectness in its traditional, graph-theoretical sense, where it means that each player is matched.) \textsc{PRA} starts off coarsening, and then refining the preferences of players according to certain rules. Then, it calls the oracle to compute a Perfect Partition with the refined preferences, which is shown to be Pareto optimal in the original instance. Here we give a much more direct interpretation of serial dictatorship, focusing only on the specific problem variant we discuss in Theorems~\ref{th:3SCB} and~\ref{th:3SCW}. Our interpretations are then illustrated in Example~\ref{ex:easy}.

\begin{theorem}[Finding a \textnormal{\textsc{poa}}, strict and complete lists, best roommate counts, 3-person rooms]
\label{th:3SCB}
If lists are strict and complete, all rooms are of capacity~3, and the best roommate counts, then serial dictatorship delivers a \textsc{poa}.
\end{theorem}

\begin{proof}
The exact implementation of serial dictatorship works in rounds, as follows. The first dictator points at her first choice. We fix this pair, and immediately assign them their final room $R_1$, which will be completed by the third player later. 
The same happens in each round: the current dictator points at her first choice \emph{available} roommate, we fix this pair, and assign them a room. If one player in the new pair is already in a room, the other one joins her, otherwise a new room is opened for them.

For dictator $i$, player $j$ is available, if both of the following hold.
\begin{enumerate}
    \item The number of roommates already assigned to a room together with either $i$ or $j$ is at most one in total, or $i$ and $j$ are already assigned to the same room.
    \item If there is no further room to open, then $j$ must be a player already assigned to a room.
\end{enumerate}

We will now show that this procedure indeed delivers a \textsc{poa}. Assume indirectly that an assignment $R' = R'_1, R'_2, \ldots, R'_k$ Pareto dominates the outcome $R = R_1, R_2, \ldots, R_k$ of our algorithm. Among the players who are better off in $R'$, we choose the one who comes earliest in the order of dictators. Let this player be~$i$. We know that $j$, the best roommate of $i$ in $R'$, is strictly better than $i$'s best roommate in~$R$. In $i$'s turn in our algorithm, $j$ thus was not an available to~$i$. The reason for this must be one of the above two points.

If the first point was the reason, then some of the earlier dictators already reserved $i$ and $j$ for themselves. Since we assumed that $i$ is the earliest dictator who is better off in $R'$ than in $R$, the choices of all earlier dictators must stay intact in $R'$, otherwise some of them receives a worse best roommate, which contradicts the fact that $R'$ is a Pareto improvement. The second case is when all $k$ rooms had already been open as $i$ was considering to point to $j$, and neither of them had been assigned to a room yet. At least one of the fixed pairs must be split up in $R'$, if $i$ and $j$ are assigned to the same room in it. The earlier dictator in this pair must precede $i$in the order of dictators, and thus any change in her best allocated roommate must worsen her situation, which contradicts the Pareto improvement property.  % both $i$ and $j$ had been in a fixed pair when $i$ was traversing through her preference list
%Let us search for a Pareto improvement. The first dictator clearly cannot be separated from her first choice roommate, because she is her only best-choice partner and the best roommate specifies the satisfaction of a player. This pair is thus a fixed element in any assignment not worse in the Pareto sense than the current one. Subject to this, the second dictator also only has a single best choice roommate, with whom she is in a fixed triplet, and she also cannot be separated from her. The same argument applies further in the induction.%\accom{The proof is very informal. There is no iteration of the described procedure mentioned at all; however it is essential.}
\qed \end{proof}

\begin{theorem}[Finding a \textnormal{\textsc{poa}}, strict and complete lists, worst roommate counts, $r_i$-person rooms, including 3-person rooms]
\label{th:3SCW}
If lists are strict and complete, the rooms are of arbitrary capacity, and the worst roommate counts, then serial dictatorship delivers a \textsc{poa}.
\end{theorem}

\begin{proof}
If the worst roommate counts, serial dictatorship can be interpreted as follows. In each round, the dictator moves into one of the smallest available rooms of size $r_i$ with her best $r_i - 1$ choice roommates. The coalition is fixed and the room is removed from the set of available rooms.

To see correctness, we apply induction. Clearly, the price for improving the position of a dictator is to harm some previous dictator, because serial dictatorship gave her the fewest possible top choices still available on her list. Thus the output of the mechanism is a \textsc{poa}. % Clearly the dictator of each step cannot be better off in any assignment, since serial dictatorship gave her the fewest possible top choices on her list. \tfcom{}
\qed \end{proof}

\begin{example}
Figure~\ref{fi:sd} illustrates an instance. We run serial dictatorship in the three settings in which it delivers a \textsc{poa}. We assume the order of dictators to be 1, 2, \dots, 9.

\begin{figure}[htbp]
\centering
	\begin{minipage}{0.3\textwidth}
		\begin{tabular}{rcccccccc}
			1:&5&4&7&3&9&6&8&2\\
			2:&1&4&5&9&8&6&3&7\\
			3:&2&5&4&9&1&6&7&8\\
		\end{tabular}
	\end{minipage}
	\begin{minipage}{0.3\textwidth}
		\begin{tabular}{rcccccccc}
			4:&3&6&7&2&9&6&8&1\\
			5:&3&6&2&7&8&4&1&9\\
			6:&7&2&8&5&4&9&1&3\\
		\end{tabular}
	\end{minipage}
	\begin{minipage}{0.3\textwidth}
		\begin{tabular}{rcccccccc}
			7:&1&2&9&3&4&6&8&5\\
			8:&6&3&7&1&9&5&4&2\\
			9:&2&4&1&6&7&3&8&5\\
		\end{tabular}
	\end{minipage}
	\caption{An instance with 9 players and strictly ordered complete lists.}
	\label{fi:sd}
\end{figure}

\begin{itemize}
\item \textbf{strict and complete lists, best roommate counts, 3-person rooms} (Theorem~\ref{th:3SCB})\\
	The first dictator, player~1 chooses her first choice partner player~5 and they become a fixed pair. Then, the second dictator, player~2 chooses player~1, so these three form a fixed triplet. The third dictator cannot choose her first or second choice players~2 and~5, thus she becomes a fixed pair with player~4. Now it is exactly player~4 who is next to choose, but since she is already coupled up with her first choice roommate player~3, she does not change the current assignment. Player~5 is already in a fixed room. Player~6 opens the third room together with player~7, who then adds player~9 to this room, because her first two choices are already taken. Player~8 then joins the room of player~3. The outcome is the following partition: $(1~2~5), (3~4~8), (6~7~9)$.
	\item \textbf{strict and complete lists, worst roommate counts, 3-person rooms} (Theorem~\ref{th:3SCW})\\
	The first dictator, player~1 chooses her first  and second choice partners player~5 and~4, and they occupy a room. The next dictator is player~2, whose first 3 choices are taken, thus she moves into a room with her two best available choices, players~9 and~8. The remaining 3 players are assigned to the last room. The outcome is the following partition: $(1~4~5), (2~8~9), (3~6~7)$.
	\item \textbf{strict and complete lists, worst roommate counts, $r_i$-person rooms} (Theorem~\ref{th:3SCW})\\
	Let the rooms have capacity 2, 3, and 4, respectively. The first dictator, player~1 chooses her first choice partner player 5, and they occupy the smallest room. The next dictator is player~2, who cannot choose her first choice player~1, since her assignment is already fixed, thus she moves into the 3-person room with her best available choices, players~4 and~9. The remaining 4 players are assigned to the largest room. The outcome is the following partition: $(1~5), (2~4~9), (3~6~7~8)$.
\end{itemize}
\label{ex:easy}
\end{example}

\subsection{Hard cases}%\accom{-. Hardness proofs in Section 4.2. The informal argument preceding the actual hardness proofs in this section seem to depend on the concept of a Turing reduction (rather than on the concept of a many-one reduction). The class NP, however, is not closed under Turing reductions. It therefore seems that the authors cannot rely on this argument and say that it suffices to prove equivalence of the NP-complete problems and showing that there is a perfect outcome that is preferred by all players. It seems, however, that recourse can be taken directly to Aziz et al [2], whose Preference Refinement Algorithm guarantees that that such a proof does suffice.}

In all remaining cases, finding a \textsc{poa} is computationally infeasible, even though it is guaranteed to exist. We show this in two steps, similarly to the technique used by Aziz et al.~\cite{ABH13}. First we observe that either all \textsc{poa}s of an instance or none of them satisfy the property that all players receive one of their best outcomes. Then, we show that an $\NP$-complete problem can be reduced to the decision problem of answering whether a \textsc{poa} exists with this property. From this follows that by computing \emph{any} \textsc{poa} in polynomial time, one could answer the $\NP$-complete decision problem in polynomial time. Below we give three proofs for three different settings. %In each of these proofs we show that an $\NP$-complete problem can be reduced to deciding whether a \textsc{poa} exists that assigns each player one of her best outcomes.

%In the coming three proofs for different cases, we define what a best outcome for a player is, and use three different $\NP$-complete decision problems, each of which naturally reduces to one of our hard problems.

\begin{theorem}[Finding a \textnormal{\textsc{poa}}, strict and complete lists, best roommate counts, $r_i$-person rooms]
\label{th:iSCB}
If lists are strict and complete, the rooms are of arbitrary capacity, and the best roommate counts, then computing a \textsc{poa} is at least a hard as the $\NP$-complete \textsc{unary bin packing} problem.
\end{theorem}

\begin{proof}%\accom{The reduction is not fully specified. What happens if an instance of unary bin packing contains items valued to one? Then, the construction of the best choices of the players cannot be made in a circular manner. If one chooses the best choices of the players representing one-valued items carelessly, then statement "all POAs order each player to a room with a first ranked roommate if and only if a bin packing exists" is not true anymore. I think this ambiguity might be seen for an instance of unary bin packing where the set of items is {2,2,1,1}.}
We construct an instance of the \textsc{poa} problem to each input of \textsc{unary bin packing} with item size at least~2. We will show that all \textsc{poa}s order each player to a room with a first ranked roommate if and only if a bin packing exists. 

Each item of size $i\geq 2$ (recall Remark~\ref{re:bin}) in \textsc{unary bin packing} corresponds to $i$ players in the \textsc{poa} problem. The players of one item have their unique first choice player among themselves, in a circular manner. The rest of the preference lists can be chosen arbitrarily. To guarantee that all players have a first choice roommate, we need to keep every preference cycle together. This is equivalent to keeping the items of the bin packing problem unsplit, and thus it is possible if and only if there is a perfect bin packing.

This reduction from \textsc{unary bin packing} with bin size $b$ immediately implies that the proof is valid even if all rooms are of equal size~$b$. 
\qed \end{proof}

\begin{theorem}[Finding a \textnormal{\textsc{poa}}, ties, complete lists, best roommate counts, 3-person rooms]
\label{th:3TCB}
If lists are weakly ordered and complete, the rooms are of capacity~3, and the best roommate counts, then computing a \textsc{poa} is at least a hard as the $\NP$-complete \textsc{directed triangle cover} problem.
\end{theorem}

\begin{proof}
We construct an instance of the \textsc{poa} problem to each input of \textsc{directed triangle cover}. Let us consider a digraph $D$, where all vertices have at least one outgoing and at least one incoming edge, as stated in Remark~\ref{re:triangle}. Vertices in $D$ correspond to players in our assignment problem. If a player $i$ had an outgoing edge towards player $j$ in $D$, then $i$ ranks $j$ first. All other players are ranked second by~$i$.  %Each outgoing edge of a vertex will carry rank~1, while all other edges will carry rank~2 in our \textsc{poa} problem. 
A directed triangle cover exists in $D$ if and only if there is an assignment where each player has at least one first choice roommate. This latter happens if and only if all \textsc{poa}s have this property.
\qed \end{proof}

\begin{theorem}[Finding a \textnormal{\textsc{poa}}, ties, complete lists, worst roommate counts, 3-person rooms]
\label{th:3TCW}
If lists are weakly ordered and complete, the rooms are of capacity~3, and the worst roommate counts, then computing a \textsc{poa} is at least a hard as the $\NP$-complete \textsc{triangle cover} problem.
\end{theorem}

\begin{proof}%\accom{This reduction also does not consider all possible cases. What happens if there is a vertex with no outgoing arc in an input instance? How does the statement "a directed triangle cover exists in D if and only if there is an assignment where each player has at least one first choice roommate" work for the following instance of directed triangle cover: a->e, b->d, c->d, f->a, e->a, e->b, d->b, d->c. This instance has neither a directed triangle (by a directed triangle I mean a directed cycle of length three) nor triangles in an underlying undirected graph; for example, vertex a cannot be in any triangle. I think that according to the reduction, we obtain an instance of finding POA in which coalitions {a, e, f} and {b, c, d} give a requested solution, that is, where each player has at least one first-choice roommate. Player a has its best roommate e in the coalition, player e has its best roommate a in the coalition, and also player f has its best roommate a in the coalition. It works for the second coalition as well.}
To each input $G$ of the \textsc{triangle cover} problem, we now construct an instance of the \textsc{poa} problem. Remember that according to Remark~\ref{re:triangle}, $G$ has no isolated or degree 1 vertex. Starting with $G$, let us assign rank~1 to all neighbors of each player. Now we complete $G$ by adding all missing edges and assign rank 2 on both end vertices of such an edge.

We claim that there is a \textsc{poa} that gives every player two of her first ranked roommates if and only if a triangle cover exists in~$G$. If a triangle cover exists in $G$, then it delivers an assignment consisting of original edges only, thus it is possible to assign each player into a room with first-choice roommates only. Since each player reaches her best outcome in this assignment, it also must be Pareto optimal. To see the other direction, we assume that there is a \textsc{poa} that orders each player to a room with only first ranked roommates. This assignment must then consist of the edges of $G$ exclusively, and thus it is a triangle cover in~$G$.
\qed \end{proof}

\section{Conclusion}

We have studied complexity issues in the Pareto optimal coalition formation problem in which players have preferences over each other, and the coalitions must be of a fixed size. We have investigated a number of variants of this problem and determined the complexity of verifying Pareto optimality, deciding the existence of a \textsc{poa}, and finding a \textsc{poa}.

One natural direction of future research is to forgo the requirement on the perfectness of the assignment. In this case, due to the feasibility of the empty assignment and monotonicity, a \textsc{poa} trivially exists, so our question raised in Section~\ref{se:existence} about the existence of an optimal solution does not apply. However, allowing the assignment to be imperfect leads to unnatural strategies in serial dictatorship. If the worst roommate matters, then the dictator is better off choosing her single best roommate and not letting anyone else into the room, however large it is. Besides this, one needs to clarify how to deal with the option of staying alone in a large room, which was prohibited in our setting. Depending on the total capacity of rooms, and the way of calculating the value of an imperfect coalition, one can define several natural variants of the problem. We remark that our results in Sections~\ref{se:verification} and~\ref{se:finding} remain valid for imperfect assignments as well, if the total capacity of rooms equals the number of players, but some agents might be left unassigned---one needs to marginally adjust serial dictatorship, while the hardness proofs carry over.

Besides this, one might consider to investigate analogous problems with cardinal preferences instead of ordinal ones, an outside option for players, or strategic behaviour.

\bibliographystyle{abbrv}
\bibliography{mybib}
\end{document}